	\documentclass[journal, letterpaper, draftcls, onecolumn]{IEEEtran}

	\usepackage{comment}
	\usepackage{dsfont}
	\usepackage{numprint}
	\usepackage[english]{babel}
	\usepackage{graphicx}
	\usepackage{subfig}
	\usepackage{cite}
	\usepackage{xspace}
	\usepackage{color}
	\usepackage{algorithm}
	\usepackage{algpseudocode}
  \usepackage{enumerate}
  \usepackage{amsthm}
	
	\newtheorem{dft}{Definition}
	\newtheorem{prop}{Proposition}
	\newtheorem{lemma}{Lemma}
	\usepackage[normalem]{ulem}
	\usepackage{cite}

\def\pth#1{\left(#1\right)}                
\def\acc#1{\left\{#1\right\}}              
\def\cro#1{\left[#1\right]}

\def\Exp#1{\exp\cro{#1}}

\def\AND{\text{and\:}}

\def\FOR{\text{for\:}}

\def\WITH{\text{with\:}}

\newsavebox{\fminibox}
\newlength{\fminilength}

 \def\T{^\tD} \def\+{^\dagger}

\def\nequiv{\not\kern-.05em\equiv}
\def\egal{\kern-.5em=\kern-.5em}        
\def\propt{\kern-.2em\propto\kern-.2em} 
\def\wt#1{\widetilde{#1}} 

\def\intdouble{\int\kern-0.3em\int}
\def\inttriple{\int\kern-0.3em\int\kern-0.3em\int}

\def\rond#1{\overset{\kern-0.33em~_\circ}{#1}}
\def\rondit[#1]#2{\overset{\kern#1~_\circ}{#2}}

	\input{alphabet}

	\def\XS{\xspace}
	\def\eg{\textit{e.g.}\XS}
	\def\ie{\textit{i.e.}\XS}
	\def\gx{\gamma_{\xb}}
	\def\gn{\gamma_{\nb}}
	\def\alphax{\alpha_{\xb}}
	\def\betax{\beta_{\xb}}
	\def\alphan{\alpha_{\nb}}
	\def\betan{\beta_{\nb}}

	\def\thetat{\thetab^{(t)}}
	
	\def\xt{\xb^{(t)}}
	\def\xtmun{\xb^{(t-1)}}

	\definecolor{Bleu}{rgb}{0.25,0.25,1}
	\definecolor{Bleu}{rgb}{0.0,0.0,0.75}
	\definecolor{Vert}{rgb}{0.4,0.7,0.4}
	\definecolor{Vert}{rgb}{0.1,0.75,0.1}
	\definecolor{Rouge}{rgb}{0.75,0.2,0.2}
	\definecolor{Noir}{cmyk}{0,0,0,1}

\begin{document}

\title{Efficient MCMC for Joint Inversion in High-Dimensional Space}
\title{Gradient Scan Sampler : an efficient MCMC algorithm in high-dimensional space \\Application in inverse problems}
\title{Gradient Scan Gibbs Sampler : \\ an efficient high-dimensional stochastic algorithm  \\Application in inverse problems}
\title{Gradient Scan Gibbs Sampler: \\ an efficient algorithm for high-dimensional Gaussian distributions}

\author{O. F\'eron$^{*}$, F. Orieux and J.-F. Giovannelli
\thanks{O. F\'eron is with {EDF Research \& Developments}, 92140 Clamart, France and with Univ. Paris Dauphine, FiME, 75116 Paris, France, \texttt{olivier-2.feron@edf.fr}. F. Orieux is with Univ. Paris-Sud 11, L2S, UMR 8506, 91190 Gif-sur-Yvette, France, \texttt{orieux@l2s.centralesupelec.fr}. J.-F. Giovannelli is with Univ. Bordeaux, IMS, UMR 5218, F-33400 Talence, France, \texttt{Giova@IMS-Bordeaux.fr}.}
}


\maketitle

\begin{abstract}
This paper deals with Gibbs samplers that include high dimensional conditional Gaussian distributions. It  proposes an efficient algorithm that avoids the high dimensional Gaussian sampling and relies on a random excursion along a small set of directions. The algorithm is proved to converge, \ie the drawn samples are asymptotically distributed according to the target distribution.
%
%
Our main motivation is in inverse problems related to general linear observation models and their solution in a hierarchical Bayesian framework implemented through sampling algorithms. It finds direct applications in semi-blind\,/\,unsupervised methods as well as in some non-Gaussian methods. The paper provides an illustration focused on the unsupervised estimation for super-resolution methods.
%
%
%
%
%
\end{abstract}



\section{Introduction}\label{sec:introduction}

\subsection{Context and problem statement}

Gaussian distributions are common throughout signal and image processing, machine learning, statistics,\dots being convenient from both theoretical and numerical standpoints. Moreover, they are versatile enough to describe very diverse situations. Nevertheless, efficient sampling including these distributions is a cumbersome  problem in high dimensions and the current paper deals with this question.

Our main motivation here is in inverse problems~\cite{Idier08,Giovannelli15} and the methodology resorts to a hierarchical Bayesian strategy, numerically implemented through Monte-Carlo Markov Chain and more specifically the Gibbs Sampler (GS). Indeed, consider the general linear direct model $\yb=\Ab \xb + \nb$, where $\yb$, $\nb$ and $\xb$ are the observation, the noise and the unknown image and $\Ab$ is a given linear operator. Consider, again, two independent prior distributions for $\nb$ and $\xb$ that are Gaussian conditionally to a vector~$\thetab$, namely the hyperparameter vector. The estimation of both $\xb$ and $\thetab$ relies on the sampling of the joint posterior $p(\xb,\thetab|\yb)$, and this is the core question of the paper. It commonly requires the handling of the high dimensional conditional posterior $p(\xb|\thetab,\yb)$ that is Gaussian with given mean $\mb$ and precision $\Qb$.

The framework directly covers non-stationary and inhomogeneous Gaussian models for image and noise. The paper has also fallouts for non-Gaussian models based on conditionally Gaussian ones involving auxiliary\,/\,latent variables\footnote{It is based on the fact that for a couple of random variables $(U,V)$, the conditional law for $U|V$ is Gaussian and the marginal law for $U$ is non-Gaussian. A famous example is a Gaussian variable with precision under a gamma law: the resulting marginal follow a Student law.} (\eg, location or scale mixtures of Gaussian) for edge preserving~\cite{Geman84,Geman95,Giovannelli08} and for sparse signals~\cite{Tan10,Kail12}. It also includes other hierarchical models~\cite{Feron07,Ayasso10} involving labels for inversion-segmentation. This framework also includes linear variant direct models and some non-linear direct models, based on conditional linear ones, \eg bilinear or multilinear. In addition, it covers a majority of current inverse problems, \eg unsupervised \cite{Giovannelli08} and semi-blind \cite{Orieux10}, by including hyperparameters and acquisition parameters in the vector~$\thetab$.

Large scale Gaussian distributions are also useful for Internet data processing, \eg to model social networks and to develop recommender systems~\cite{Liu11}. They are also widely used in epidemiology and disease mapping~\cite{Besag72,Rue01} as they provide a simple way to include spatial correlations. The question is also in relation with spatial linear regression with (smooth) spatially varying parameters~\cite{Gelfand03}. In these cases the question of efficient sampling including Gaussian distributions in high dimensions becomes crucial and it is all the more true in the ``Big Data'' context.

\subsection{Existing approaches}

The difficulty is directly related to handling the high-dimensional precision $\Qb$. The factorization (Cholesky, square root,\dots), diagonalization and inversion of $\Qb$ could be used but they are generally infeasible in high dimensions due to both computational cost and memory footprint.
Nevertheless, such solutions are practicable in two famous cases.
\begin{itemize}
\item If $\Qb$ is circulant or circulant-block-circulant an efficient strategy~\cite{Chellappa85,Chellappa92} relies on its diagonalization computed by FFT. More generally, an efficient strategy exists if $\Qb$ is diagonalizable by a fast transform, \eg discrete cosine transform for Neumann boundary conditions~\cite{Bardsley12,Bardsley13}. 
\item When $\Qb$ is sparse, a possible strategy~\cite{Rue01,Rue05,Lalanne01} relies on a Cholesky decomposition and a linear system resolution. Another strategy is a Gibbs sampler~\cite{Winkler03} that simultaneously updates large blocks of variables.
\end{itemize}
In order to address more general cases, solutions founded on iterative algorithms for objective optimization or linear system resolution have recently been proposed.

\begin{enumerate}


%
%

\item An efficient algorithm has been proposed by several authors \cite{Tan10,Papandreou10,Orieux12, Bardsley12, Bardsley13} (previously used in applications \cite{Feron07,Orieux10}). It is founded on a Perturbation-Optimization principle: adequate stochastic perturbation of a quadratic criterion and optimization of the perturbed criterion. However, in order to obtain a sample from the right distribution, an exact optimization is needed, but practically an empirical truncation of the iterations is implemented, leading to an approximate sample. \cite{Gilavert13} introduces a Metropolis step in order to asymptotically retrieve an exact sample and then to ensure, in a global MCMC procedure, the convergence to the correct invariant distribution.
\item In \cite{Fox08,Parker12} the authors propose a Conjugate Direction Sampler (CDS) based on two crucial properties: $(i)$ a Gaussian distribution admits Gaussian conditional distributions and $(ii)$ a set of mutually conjugate directions w.r.t.~$\Qb$ is available. The key point of the algorithm is to sample along these mutually conjugate directions instead of optimize as in the classic Conjugate Gradient optimization algorithm.
\end{enumerate}

In the first case, the only constraint on $\Qb$ is that a sample from $\Nc(0,\Qb)$ must be accessible, which is often the case in inverse problem applications. In the second case, $\Qb$ must have only distinct eigenvalues to make the CDS give an exact sample. Otherwise it leads to an approximate sample as described in \cite{Parker12}.

The proposed algorithm uses the same approach as the CDS and extends the efficiency to, theoretically, any matrix $\Qb$.

\subsection{Contribution}
The existing methods described above and the proposed one are both founded on a Gibbs Sampler. However, the existing ones attempt to sample the high dimensional Gaussian component $\xb\in \Rbb^N$ whereas the proposed method does not. Our main contribution is to avoid the high dimensional sampling and only requires small dimensional ones.
More precisely, given a subset $D \subset \Rbb^N$, the keystone of the advance is to sample the sub-component of $\xb$ according to the subset $D$. It must be sampled under the appropriate conditional distribution $\pi(\xb_{D} | \xb_{\setminus D}, \thetab)$, with the decomposition $\xb=(\xb_D,\xb_{\setminus D})$. The algorithm takes advantage of the ease of calculating the conditional pdf of a multivariate Gaussian, when $D$ is appropriately built, as explained in section \ref{sec:gradient-scan-gibbs}.
These ideas are strongly related to different existing works.
\begin{itemize}
	\item If the subset $D$  is composed of only one direction in the canonical coordinates, the algorithm amounts to a pixel-by-pixel GS \cite{Geman84}.
	\item The marginal chain $\xb^{(t)}$ can also be viewed as the one produced by a specific random scan sampler~\cite{Levine05,Levine06,Latuszynski13}. The random scans are related to the random choice of $D$, depending on the current value $\thetab^{(t)}$.
	\item Other algorithms based on optimization principles \cite{Parker12,Orieux12b} aim at producing a complete optimization. On the contrary, in essence, the proposed approach only requires a few steps of the optimization process.
	\item A similar idea is at work in Hamiltonian (or Langevin) Monte Carlo \cite{Stramer99a,Stramer99b,Duane87, Neal10} (see also~\cite{Vacar11}): the proposal law takes advantage of an ascent direction of the target to increase the acceptation probability. Here, the exact distribution is sampled, so the proposal is always accepted.
\end{itemize}
%

However, to our knowledge, the proposed algorithm does not directly join the class of existing strategies. One contribution of this paper is to give sufficient assumptions for convergence, \ie the samples are asymptotically distributed according to the joint pdf $p(\xb,\thetab)$. 

\subsection{Outline}

Subsequently, Section~\ref{sec:gradient-scan-gibbs} presents the proposed algorithm and section~\ref{sec:super-resolution} gives an illustration through an academic problem in super-resolution. Section~\ref{sec:conclusion} presents conclusions and perspectives.

\section{Gradient Scan Gibbs Sampler}
\label{sec:gradient-scan-gibbs}
In this section we describe the proposed algorithm: a Gibbs sampler with a high dimensional conditional Gaussian distribution. The objective is to generate samples from a joint distribution $p(\xb,\thetab)$, where $\xb\in \Rbb^{N}$ is highly dimensional and $p(\xb|\thetab)$ is a Gaussian distribution $\Nc(\mb_\thetab,\Qb_\thetab^{-1})$:
\begin{equation}
p(\xb|\thetab)=(2\pi)^{-N/2} (\det \Qb_\thetab)^{1/2} \exp -J_{\thetab}(\xb)
\end{equation}
with the potential $J_\thetab$ defined as:
\begin{equation}
J_\thetab(\xb) = \frac{1}{2} (\xb - \mb_\thetab)\T \Qb_\thetab (\xb - \mb_\thetab).
\end{equation}
All the other variables of the problem are grouped into $\thetab \in \Theta$ and we assume that the sampling from $p(\thetab|\xb)$ is tractable (directly or with several steps of the Gibbs sampler, including Metropolis-Hastings steps).

\subsection{Preliminary results}
\label{sec_preliminary_results}

This section presents classic definitions and results, mostly based on~\cite{Fox08}, needed to provide convergence proof and  links between matrix factorization and optimization\,/\,sampling procedures. 

Consider $\Qb$ a $N \times N$ symmetric definite positive matrix.

\begin{dft}
	A set $\acc{ \db_n, n=1,\dots,N }$ of non-zero vectors in $\Rbb^N$ such that:
	%
	$
	\db_n\T \, \Qb \, \db_m = 0 \quad \FOR n,m=1,\dots,N ,~ n \neq m
	$
	%
	is said mutually conjugate w.r.t. $\Qb$. \hfill $\triangle$
\end{dft}

A mutually conjugate set $\acc{\db_1,\dots,\db_N}$ w.r.t. $\Qb$ is a basis of $\Rbb^N$, then, for all $\xb\in\Rbb^N$:
\begin{equation*}
	\xb = \sum_{n=1}^N \alpha_n \db_n ~~~\WITH~ \alpha_n = \frac{\db_n\T \Qb \xb}{\db_n\T \Qb \db_n} \,.
\end{equation*}
So, if $\xb\sim\Nc(\mb,\Qb^{-1})$ is a Gaussian random vector with mean $\mb$ and precision $\Qb$, then the $\alpha_n$ are also Gaussian:
\begin{equation}\label{Eq:LoiAlpha}
	\alpha_n \sim \Nc \left(\frac{\db_n\T \Qb \mb}{\db_n\T \Qb \db_n} ~;~\frac{1}{\db_n\T \Qb \db_n} \right)
\end{equation}
and reciprocally if the $\alpha_n$ are distributed under~\eqref{Eq:LoiAlpha} then $\xb\sim\Nc(\mb,\Qb^{-1})$.

In particular, let $\xb^0 \in \Rbb^N$ be a ``current'' point and $\db_1 \in \Rbb^N$ a given ``direction''. One can find $\db_2,\dots,\db_N$ such that $\acc{\db_1,\dots,\db_N}$ is mutually conjugate w.r.t. $\Qb$ and $\xb^0$ writes:
\begin{equation*} 
	\xb^0 = \sum_{n=1}^N \alpha_n^0 \, \db_n.
\end{equation*}
Consider now the $N_D$-dimensional subset

\begin{align*}
	D(\xb^0)
	& = \acc{ \sum_{n=1}^N \alpha_n \db_n, 
	\alpha_n \in\eR,~ n \le N_D, \alpha_n=\alpha_n^0,~ n> N_D  }\\
	& = \Big\{ \xb^0 + \sum_{n=1}^{N_D}(\alpha_n - \alpha_n^0) \db_n, ~(\alpha_1,\ldots,\alpha_{N_D})\in\eR^{N_D}  \Big\}
\end{align*}
We are interested in the conditional pdf $p(\xb |\xb \in D(\xb^0))$. The following result and its proof can be found in \cite{Fox08}.

\begin{prop}
	A sample $\wt{\xb}$ according to $p(\xb |\xb \in D(\xb^0))$ can be obtained by:
	\begin{enumerate}\setlength{\itemsep}{0pt}
		\item sample independently the set $(\wt{\alpha}_1,\ldots,\wt{\alpha}_{N_D})$ with:
		$$
		 \wt{\alpha}_n \sim \Nc \pth{ \frac{\db_n\T \Qb (\xb^0-\mb)}{\db_n\T \Qb \db_n} ~;~ \frac{1}{\db_n\T \Qb \db_n} }, n=1,\ldots,N_D
		 $$
		%
		\item compute $\wt{\xb} = \xb^0-\sum_{n=1}^{N_D}\wt{\alpha}_n \, \db_n$
	\end{enumerate}
\end{prop}


\subsection{Gradient Scan Gibbs Sampler (GSGS)}

In the following we propose a Gibbs sampling algorithm in order to sample the joint probability $p(\xb,\thetab)$. The principle is to sample, at each iteration of the Gibbs sampler, only $N_D$ directions of $\xb$ instead of sampling the whole high dimensional variable. The chosen first direction of the set $D$ will be the gradient of the potential of $p(\xb|\thetab)$, with a stochastic perturbation to ensure, in the general case, the convergence of the resulting Markov chain. The following directions are chosen so as to get a mutually conjugate subset with respect to the precision of $p(\xb|\thetab)$.

We call our proposed algorithm the Gradient Scan Gibbs Sampler (GSGS) which is described by Algorithm \ref{algo_pgsgs}.
\begin{algorithm}[htb]

 \caption{: Gradient scan Gibbs sampler (GSGS).\label{algo_pgsgs}}

Define an initial point $\xb^{(0)}$, a number $N_D$ and a stopping criterion. Iterate \vspace*{2mm}.

\begin{algorithmic}[1]

   \State sample $\thetat \sim p(\thetab | \xtmun)$ \vspace*{3mm}

   \State set $\Qb_t=\Qb_{\thetat}$ and $\mb_t=\mb_{\thetat}$, and compute the gradient $\gb=\nabla J_\thetab(\xtmun) = \Qb_t(\xtmun - \mb_t)$\vspace*{2mm}

   \State sample a perturbation $\tilde{\varepsilonb} \sim p(\varepsilonb)$ \vspace*{2mm}

   \State compute a set of $N_D$ mutually conjugate directions $(\db_1,\ldots,\db_{N_D})$ w.r.t. $\Qb_t$ such that
   $$\db_1 = \gb + \tilde{\varepsilonb} $$


   \State sample independently the set $(\wt{\alpha}_1,\ldots,\wt{\alpha}_{N_D})$ with:
   $$
   \wt{\alpha}_n \sim \Nc \pth{ \frac{\db_n\T \gb}{\db_n\T \Qb_t \db_n} ~;~ \frac{1}{\db_n\T \Qb_t \db_n} }, ~ n\le N_D
   $$


   \State compute $\xt = \xtmun -\sum_{n=1}^{N_D}\wt{\alpha}_n \, \db_n$ \vspace*{2mm}

	\State $t\gets t+1$.

\end{algorithmic}
until the stopping criterion is reached.
\end{algorithm}
In this algorithm the chosen first sampling direction $\db_1$ is given by the gradient of the potential of $p(\xb|\thetab)$, with an additional random perturbation $\tilde{\varepsilonb}$ that follows a probability density $p(\varepsilonb)$. In fact, we expect the gradient to be a good direction towards regions of high probabilities. Also, the gradient is easily computable and so gives an easy rule to sample from any current point $\xb$. Moreover, the other conjugate directions are iteratively computable as described in the Conjugate Direction Sampling (CDS) algorithm\cite{Fox08} used to get an approximated sample from a Gaussian distribution. In fact, the GSGS is embedding steps of the CDS in a global Gibbs sampler.

The objective is now to study the convergence properties of the GSGS. We begin with two classic results.
\begin{itemize}
\item If the Markov chain is aperiodic, $\phi-$irreducible for some nonzero measure $\phi$\footnote{In all the paper we will consider $\phi$ as the Lebesgue measure and we will omit it for simplicity.}, and has an invariant probability $\pi$, then it converges to $\pi$ from $\pi$-almost every starting point (cf. Theorem 4.4 of \cite{Gilks96}).
\item Moreover, if the Markov chain is Harris recurrent, then it converges to $\pi$ from all starting point \cite{Gilks96, Roberts06a}.
\end{itemize}
The Harris recurrence of Gibbs samplers, or more generally Metropolis-within-Gibbs samplers is well studied in \cite{Roberts06a}. In particular, the Theorem 12 and Corollary 13 of \cite{Roberts06a} ensures that if the Markov chain produced by the GSGS is irreducible then it is Harris recurrent. Consequently, in the following we focus on showing that the Markov chain is aperiodic, irreducible and with stationary distribution $p(\xb,\thetab)$. \\

It is trivial to see that the Markov chain $(\xt,\thetat)_{t\ge 0}$, produced by the GSGS, is aperiodic since for any non-negligible subset $A \in \Rbb^N$ including $\xtmun$, $\Pbb(\xt \in A)>0$. The existence of an invariant probability and the irreducibility can be shown by thinking of a random scan Gibbs sampling for the marginal component $(\xt)_{t\ge 0}$.
\begin{prop}
\label{prop_proba_invariante}
The Markov chain produced by Algorithm~\ref{algo_pgsgs} admits $p(\xb,\thetab)$ as an invariant distribution, even without perturbations of the gradient direction (\ie $\tilde{\varepsilonb}=0$). \\
Moreover, if the density $p(\varepsilonb)$ is supported on $\Rbb^N$, the Markov chain produced by Algorithm \ref{algo_pgsgs} is irreducible, and therefore its law converges to $p(\xb,\thetab)$.
\end{prop}
\begin{proof}
see appendix \ref{annexe_proba_invariante}.
\end{proof}
The Proposition \ref{prop_proba_invariante} then shows that the joint probability $p(\xb,\thetab)$ remains an invariant distribution in the limit case where the first direction  $\db_1$ is exactly the gradient of $p(\xb|\thetab)$, without random perturbation. However the perturbation is needed to ensure the irreducibility (and then the convergence) of the chain.
	%
	%

 If the gradient is not perturbed, the mutually conjugate set $D$ is then given by a deterministic function of $\thetat$ and $\xtmun$. In this case, we need more assumptions to ensure the Markov chain to be irreducible. For example, we can have the following result.
\begin{prop}
\label{prop_irreductibilite}
Suppose the following conditions are satisfied:
\begin{enumerate}[{H}-1]
\item The function $\thetab \mapsto \Qb_\thetab$ is continuous \label{cond_continuite}
\item $\forall (\xb,\thetab)\in \Rbb^N \times \Theta$ and $\forall r>0$, $\Pbb (\Bc(\thetab,r)|\xb)>0$, with $\Bc(\thetab,r)$ 
the ball in $\Theta$, centered in $\thetab$, of radius $r$. \label{cond_theta}
\item $\forall \xb \in \Rbb^N$, $\exists \thetab \in \Theta$ such as:
\begin{enumerate}[{H-3}.1]
\item $\Qb_\thetab$ has $N$ distinct eigenvalues, \label{cond_eigval}
\item $\xb-\mb_\thetab$ is not orthogonal to any eigenvector of~$\Qb_\thetab$,\label{cond_eigvec}
\end{enumerate}
\end{enumerate}
Then the Markov chain produced by Algorithm \ref{algo_pgsgs} without the perturbation step 3 ($\tilde{\varepsilonb}=0$) is irreducible.
\end{prop}
\begin{proof}
see appendix \ref{annexe_irreductibilite}
\end{proof}
The conditions described in Proposition \ref{prop_irreductibilite} are very restrictive and, in particular, condition H-3.1 is difficult, if not impossible, to prove in practice. This condition ensures that every non-negligible subset of $\Rbb^N$ can be reached with a non-zero probability. It can be interpreted in the framework of Krylov spaces as in \cite{Parker12}. For example, if there is $t$ such as the Krylov space
\begin{equation*}
\Kc^N(\Qb_{\thetat},\xt):= \text{span}\left(\xt,\Qb_{\thetat}\xt,\dots,\Qb_{\thetat}^N\xt \right)
\end{equation*}
is of rank $N$ then the Markov chain is irreducible. This condition can be weakened in our case because the Gaussian parameters $\mb_{\thetat}$ and $\Qb_{\thetat}$ are changing since $\thetab$ is changing at each iteration of the Gibbs sampler. Therefore a sufficient condition to ensure the irreducibility of the chain can be expressed as follows:
\begin{prop}
\label{prop_krylov}
If there is $T>N$ such as the union of Krylov spaces
\begin{equation*}
\cup_{t=1}^T \Kc^N(\Qb_\thetat,\xt) \cup \Kc^N(\Qb_\thetat,\mb_\thetat)
\end{equation*}
is of rank $N$ then the Markov chain built by the GSGS without perturbation of the gradient is irreducible.
\end{prop}
\begin{proof}
The condition implies that for any non-negligible subset $A \subset \Rbb^N$, $\Pbb\left(\xb^{(T)} \in A | \xb^{(0)}\right)>0$, which ensures the irreducibility.
\end{proof}
The issue of determining general conditions, as in Proposition \ref{prop_irreductibilite}, is an open problem at this time. The fact that the condition described in Proposition \ref{prop_krylov} is satisfied, highly depends on the model's characteristics. That is why the GSGS (with the random perturbation step 3) is the one that ensures, in all cases, the convergence of the Markov chain to the joint distribution $p(\xb,\thetab)$.

The presented results do not allow us to get any convergence rate of the Markov chain. The latter is, in fact, very important to ensure in practice the efficiency of the estimators produced by simulations in finite time. In particular, the geometric ergodicity \cite{Meyn93} is a very well known property that gives a Central Limit Theorem and ensures the Markov chain to quickly converge and give estimations of standard errors. However the Algorithm \ref{algo_pgsgs} aims to be general while the precise study of geometric convergence (especially to quantify the convergence rate) would need to specify the distributions on the parameters $\thetab$ and on the perturbation $\varepsilonb$. At this time, only weak assumptions are considered on these probabilities and the next section discusses about the different choices of $p(\varepsilonb)$ from a feasibility point of view.


\subsection{Choice of $p(\varepsilonb)$}
\label{sec:choice_epsilon}
As previously specified, the only condition to ensure the convergence of the GSGS in the general case, is to choose a distribution $p(\varepsilonb)$ supported in $\Rbb^N$. In practice we also expect a sample from $p(\varepsilonb)$ to be easily accessible. A natural choice is the Gaussian iid distribution $\Nc(0,\Ib_N)$, $\Ib_N$ being the $N \times N$ identity matrix. This was already studied in \cite{Orieux15} in the case of only sampling from a Gaussian distribution $p(\xb)$ and where results are shown in small dimensions. 

Our empirical studies in high dimension (one example is shown in section \ref{sec:super-resolution}) incited us to choose the Gaussian distribution $\Nc(0,\Qb_\thetab)$, when it is possible. The sampling from this distribution may actually be easily computable, provided that $\Qb_\thetab$ has, for example, the specific factorization form described in \cite{Orieux12b}:
\begin{equation*}
\Qb_\thetab = \sum_{k=1}^K \Mb_k\T \Rb_k^{-1} \Mb_k
\end{equation*}
In this case, the sampling from $\Nc(0,\Qb_\thetab)$ is easily computable by using the Perturbation Optimization (PO) algorithm \cite{Orieux12b}. The latter consists in (i) randomly modifying the potential $J_\thetab(\xb)$ to get a perturbed potential $\tilde{J}_\thetab$ and (ii) optimizing $\tilde{J}_\thetab$. The first step of this optimization procedure consists in computing the gradient $\nabla \tilde{J}_\thetab$ and it is trivial to show that it can be decomposed: $\nabla \tilde{J}_\thetab(\xb)=\nabla J_\thetab(\xb)+\varepsilonb$, with $\varepsilonb \sim \Nc(0,\Qb_\thetab)$. Therefore, the perturbed gradient $\db_1$ of the GSGS, with a random perturbation $\varepsilonb \sim \Nc(0,\Qb_\thetab)$, can be obtained by using the PO algorithm truncated to one step of the optimization procedure. \\
\indent
Although this choice is empirical, at this time, we may propose some intuition to recommend, when it is possible, the distribution $\Nc(0,\Qb_\thetab)$.
%
%
The first direction $\db_1$ is related to the gradient of $J_\thetab$, in accordance with the objective to get a direction towards regions of high probability. This gradient is mostly driven by the highest eigenvalues of $\Qb_\thetab$. The perturbation $\varepsilonb$ is only needed to ensure the GSGS convergence, but the objective is to keep a direction towards high probability regions. The sampling from $\Nc(0,\Qb_\thetab)$ seems to be a good compromise: it gives values of $\varepsilonb$ mostly driven by the highest eigenvalues of $\Qb_\thetab$ and then the resulting direction $\db_1$ still continues to encourage the exploration space of high probability.

We may also notice that some relaxations of the GSGS are possible, following classic arguments of a random scan Gibbs sampling. For example, it is not necessary to sample the perturbation from $p(\varepsilonb)$ at each iteration, it is sufficient to do this an infinite number of times to ensure the chain to be irreducible\footnote{From any point $\left( \xt,\thetat \right)$, let $s>t$ be the closest next time where $\varepsilonb$ is sampled, then for any non-negligible subset $A\in\Rbb^N \times \Thetab$, we have $P(\xt,A)>0$.}. As we will see in section \ref{sec:super-resolution}, a low frequency sampling of $\varepsilonb$ can improve the algorithm's efficiency.

\section{Unsupervised super resolution as a large scale problem}\label{sec:super-resolution}

\subsection{Problem statement}\label{sec:SR-Statement}

The paper details an application of the proposed GSGS to a super-resolution problem (identical to the one presented in~\cite{Orieux12b, Gilavert15}): several blurred, noisy and down-sampled (low resolution) observations of a scene are available to retrieve the original (high resolution) scene \cite{Park03,Rochefort06}.


The usual direct model reads: $\yb=\Ab\xb+\nb=\Sb\Hb\xb+\nb$. In this equation, $\yb \in \eR^M$ collects the pixels of the low resolution images (five $128 \times 128$ images, \ie $M=81920$) and $\xb \in \eR^N$ collects the pixels of the original image (one $256 \times 256$ image, \ie $N=65536$). The noise $\nb \in \eR^M$ accounts for measurement and modeling errors. $\Hb$ is a $N\times N$ circulant-block-circulant convolution matrix accounting for the optical and the sensor parts of the observation system. Here it is a square window of 5-pixel-width. $\Sb$ is a $M\times N$ matrix modeling motion (here translation) and decimation: it is a down-sampling binary matrix indicating which pixel of the blurred image is observed.


The chosen prior for the noise is $\nb \sim \Nc(\mathbf{0},\gn^{-1}\Ib)$, \ie uncorrelated. Regarding the object, the chosen prior accounts for smoothness:  $\xb \sim \Nc(\mathbf{0},\gx^{ -1} \Db\T\Db)$ where $\Db$ is the $N\times N$ circulant convolution matrix of the Laplacian filter. The hyperparameters $\gn$ and $\gx$ are unknown and the assigned priors are conjugate : Gamma distributions $\gn \sim \mathcal{G} \left(\alphan ; \betan \right)$ and  $\gx \sim \mathcal{G} \left(\alphax ; \betax \right)$. They are poorly informative for large variances and uninformative Jeffreys' prior when the $(\alphax,\betax)$ tends to $(0,0)$.
As a consequence, the full posterior pdf writes
\begin{align}
	p(\xb, \gx, \gn | \yb)
	& \propto  p(\yb | \xb, \gn) p(\xb | \gx) p(\gx) p(\gn)  \label{eq:1} \\
	& \propto \gn^{\alphan+N/2-1} \gx^{\alphax+(M-1)/2-1} \nonumber \\
	& \phantom{\propto}\ \Exp{-\gn \|\yb - \Sb\Hb\xb\|^2 /2} \Exp{- \betan\gn } \nonumber \\
	& \phantom{\propto}\ \Exp{- \gx\|\Db \xb\|^2 /2} \Exp{- \betax\gx }. \nonumber
\end{align}
The conditional law of the image writes
\begin{equation*}
	p(\xb| \yb, \gx, \gn) \propto \Exp{-\frac{\gn}{2} \|\yb - \Sb\Hb\xb\|^2 - \frac{\gx}{2} \|\Db \xb\|^2}.
\end{equation*}
Accordingly the negative logarithm gives the criterion
\begin{equation*}
	J_{\gx,\gn}(\xb) = \frac{\gn}{2} \|\yb - \Ab\xb\|^2 + \frac{\gx}{2} \|\Db \xb\|^2
\end{equation*}
and the gradient
\begin{align*}
	\nabla J_{\gx,\gn}(\xb) & = \gn \Ab\T(\Ab\xb -  \yb) + \gx \Db\T\Db \xb \\
	& = \Qb(\xb - \mb)
\end{align*}
with $\mb = \gn\Ab\T \yb$, and the Hessian
\begin{align*}
	\Qb_{\gx,\gn} = \nabla^2 J_{\gx,\gn}(\xb) & = \gn \Ab^t \Ab + \gx \Db\T\Db
\end{align*}

\subsection{Gibbs sampler}\label{sec:SR-Sampler}

The posterior pdf is explored by the proposed Gibbs sampler in
Algorithm~\ref{algo_sr}, based on the GSGS, that
iteratively updates $\gn$, $\gx$ and a subset of $\xb$. Regarding the hyperparameters, the conditional pdf are Gamma and their parameters are easy to compute.

\begin{algorithm}[htb]

	\caption{: GSGS for super-resolution. \label{algo_sr}}

	Set $t=1$, define an initial point $\xb^{(0)}$, and repeat

	\begin{algorithmic}[1]

		\State Sample $\gn^{(t)} \sim p\left(\gn | \yb, \xb^{(t-1)}\right)$
		as
		\begin{equation*}
			\mathcal{G}\left(\frac{N}{2}; \frac{2}{\|\yb - \Sb\Hb\xb^{(t-1)}\|^2} \right).
		\end{equation*}

		and $\gx^{(t)} \sim p\left(\gx | \yb, \xb^{(t-1)}\right)$
		as
		\begin{equation*}
			\mathcal{G}\left(\frac{M-1}{2}; \frac{2}{\|\Db\xb^{(t-1)}\|^2} \right).
		\end{equation*}

		\State Set $\Qb_t = \Qb_{\gx^{(t)},\gn^{(t)}} $ and compute the gradient
		\begin{equation*}
			\gb^{(t)}  = \nabla J_{\gx,\gn}\left(\xb^{(t-1)}\right) =
			\Qb_t (\xb^{(t-1)} - \mb)
		\end{equation*}

		\State \label{item:1} Sample a perturbation
		$\varepsilonb^{(t)} \sim \Nc(0,\Qb_t)$

		\State \label{item:4} Compute a set of $N_D$ mutually conjugate directions
		$\{\db_1, \dots, \db_{N_D}\}$ with the first being
		$\db_1 = \gb^{(t)} + \varepsilonb^{(t)}$.

		\label{item:2}
		\State Sample independently the set 
$(\wt{\alpha}_n)_{n=1,\ldots,N_D}$ with:
		\begin{equation*}
			\wt{\alpha}_n \sim \mathcal{N}\left( \frac{\db_n \gb^{(t)}}
			{\db_n \Qb_t \db_n};
			\frac{1}{\db_n \Qb_t \db_n} \right)
		\end{equation*}
		\State Compute
			$\xb^{(t)} \gets \xb^{(t-1)} -\sum_{n=1}^{N_D}\wt{\alpha}_n \, \db_n.$

		\State $t\gets t+1$.

	\end{algorithmic}
	until the stopping criterion is reached.

\end{algorithm}

The set of mutually conjugate directions w.r.t. $\Qb_{\gx,\gn}$, at step~\ref{item:4} of
Algorithm~\ref{algo_sr}, is computed by the Gram-Schmidt process applied
to gradient, as usually found in conjugated gradient optimization
algorithm. The procedure is similar to the algorithm described in
\cite{Parker12}. Finally the estimator is the posterior mean computed
as the empirical mean of the samples.

Despite the convergence proof with almost any law for the perturbation $\varepsilonb$ (provided that the density $p(\varepsilonb)$ is supported in $\Rbb^N$), some tuning is necessary to practically obtain a good
space's exploration. In practice, the step~\ref{item:1} has a major
influence and, as already discussed in section \ref{sec:choice_epsilon}, we observe that a working perturbation corresponds to those
of the PO algorithm \cite{Orieux12b}
\begin{equation*}
	\varepsilonb^{(t)} = {\gn^{(t)}}^{-1/2} \Ab\T \varepsilonb_{\nb} +
	{\gx^{(t)}}^{-1/2} \Db\T \varepsilonb_{\xb}
\end{equation*}
where $\varepsilonb_{\times}$ are two Gaussian normalized random vectors, leading to a Gaussian perturbation $\varepsilonb^{(t)}$ of covariance $\Qb_t$.  
However, the proposed algorithm has numerous advantages over the PO
algorithm. First the proposed algorithm has a convergence proof
because it does not suffer from truncation, even in the extreme case with $N_D=1$. Second the
perturbation has the sole constraint of having $\mathds{R}^N$ as
support. Moreover a perturbation is not required at each iteration. 

\subsection{Numerical results}\label{sec:R-SResults}

The posterior law (\ref{eq:1}) has been explored with the following
four algorithms or settings.
\begin{itemize}
\item The adaptive RJ-PO algorithm \cite{Gilavert15}, directly tuned with the acceptance probability, here chosen to be 0.9. This acceptance probability leads to an average
    number of around 150 iterations of conjugate gradients to compute the proposal, and with 6\% of rejected samples.
\item Algorithm~\ref{algo_sr} with $N_D = 150$. The idea is to build
  an algorithm close to RJ-PO's computing time.
\item Algorithm~\ref{algo_sr} with $N_D = 10$. The idea is to show
  that our algorithm offers the possibility to reduce the number of
  iterations while still offering a good exploration and with
  guaranteed convergence. We empirically found that $N_D=7$ is
    the lower limit case to have a good global
    exploration including the hyperparameters.
\item Algorithm~\ref{algo_sr} with $N_D = 2$. The idea is to show a
  very fast algorithm that offers a partially correct
  exploration. This case is particular in the sense that the
  perturbation is done only once for the whole algorithm.
\end{itemize}

The posterior mean estimations (\textsc{pm}) of the high-resolution
image are given in Fig.~\ref{fig:n1} as well as the posterior
standard deviation (\textsc{psd}). From these results we can say that
all algorithms provide similar quality for the image estimation. The same
statement can be made for the standard deviation. However the
posterior standard deviation with $N_D = 2$ seems
incorrect. A possible interpretation is that the perturbation vector $\varepsilonb$ is simulated only once during the whole algorithm. Thus, the space is surely not sufficiently explored and the covariance estimation is severely biased. Indeed, since $\varepsilonb_{\times}$ are drawn only once, the stochastic explorations are limited to the conjugate direction plus the
two directions $\varepsilonb_{\xb}$ and $\varepsilonb_{\nb}$. However the
mean estimation does not seem to be affected and this algorithm is
able to provide very quickly a good estimation of the image and hyperparameters
value.

\begin{figure*}[htbp]
	\centering

	\subfloat[RJ-PO \textsc{pm}]{
		\includegraphics[width=0.23\textwidth]{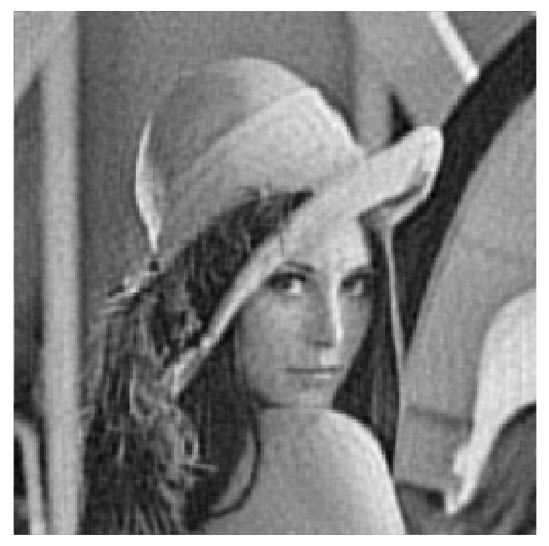}
	} %
	\subfloat[GSGS \textsc{pm} $N_D=150$]{
		\includegraphics[width=0.23\textwidth]{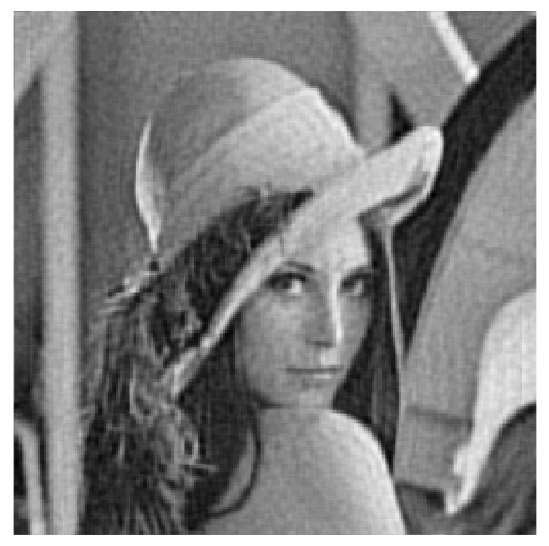}
	} %
	\subfloat[GSGS \textsc{pm} $N_D=10$]{
		\includegraphics[width=0.23\textwidth]{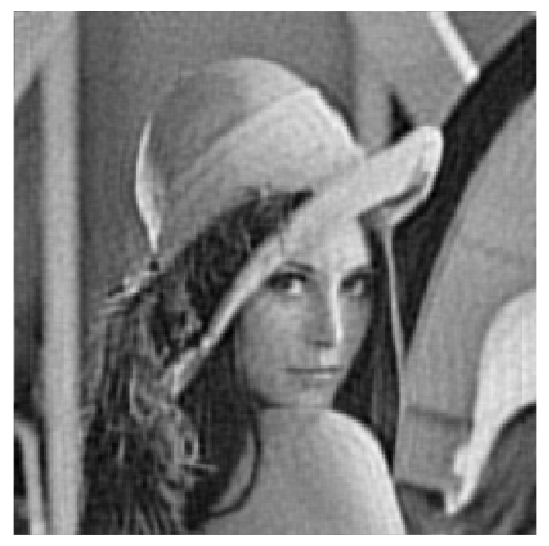}
	} %
	\subfloat[GSGS \textsc{pm} $N_D=2$]{
		\includegraphics[width=0.23\textwidth]{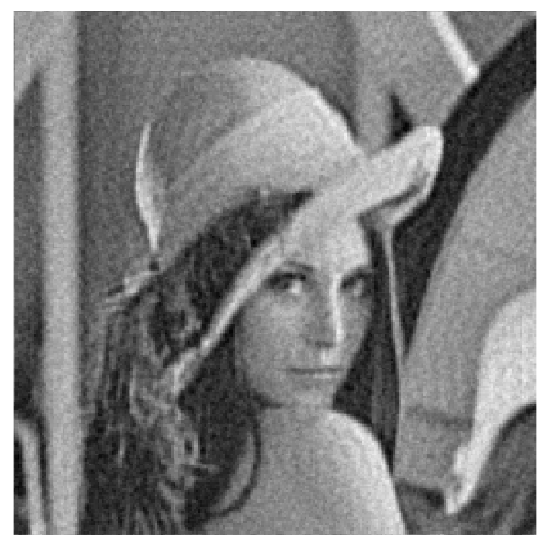}
	} %

	\subfloat[RJ-PO \textsc{psd}]{
		\includegraphics[width=0.23\textwidth]{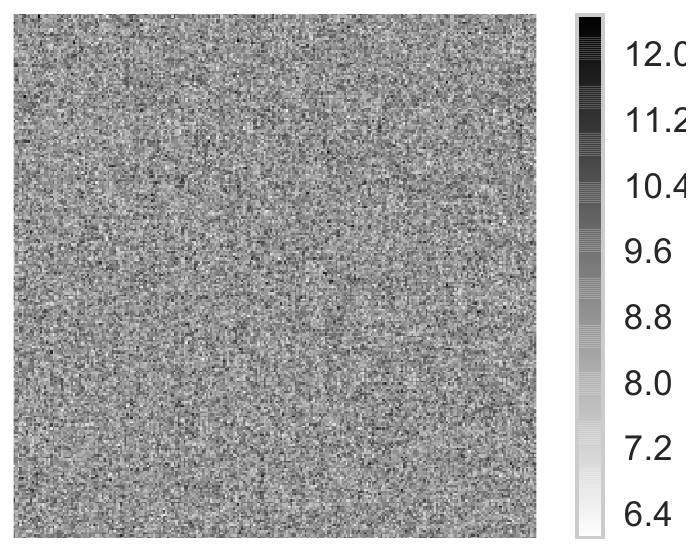}
	} %
	\subfloat[GSGS \textsc{psd} $N_D=150$]{
		\includegraphics[width=0.23\textwidth]{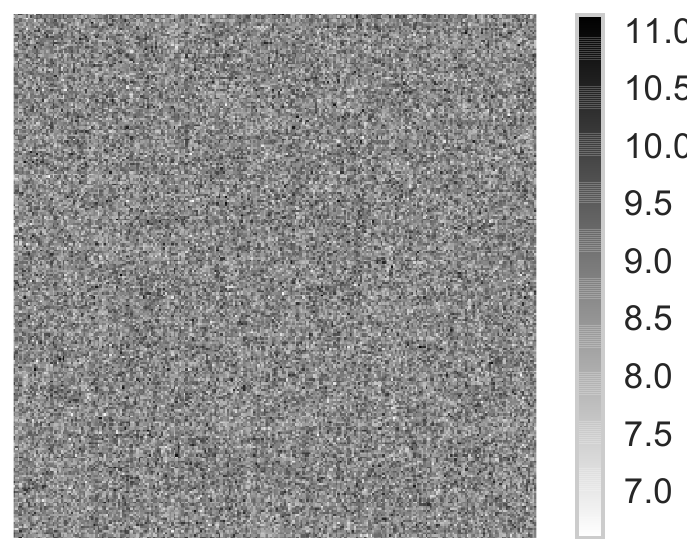}
	} %
	\subfloat[GSGS \textsc{psd} $N_D=10$]{
		\includegraphics[width=0.23\textwidth]{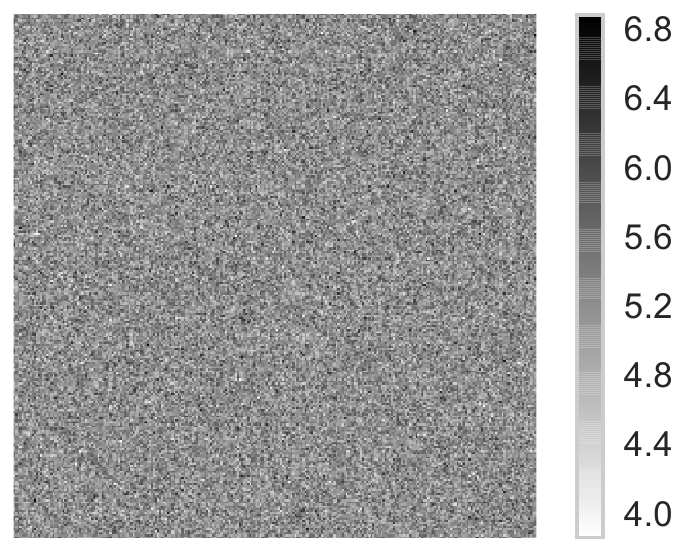}
	} %
	\subfloat[GSGS \textsc{psd} $N_D=2$]{
		\includegraphics[width=0.23\textwidth]{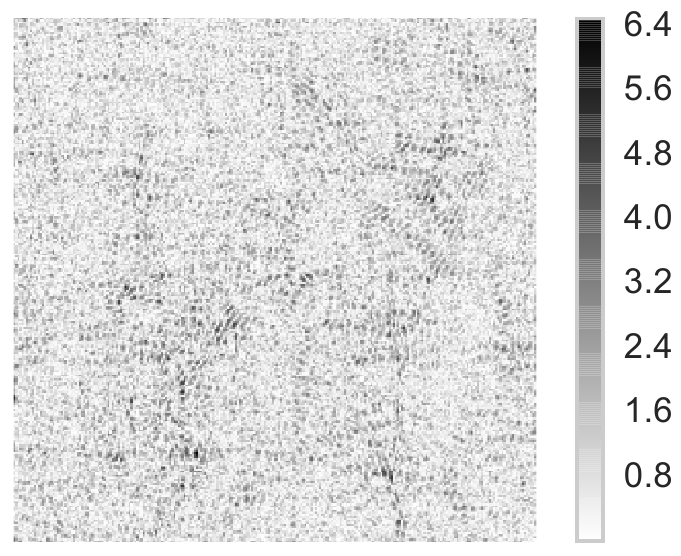}
	}

	\caption{Image results.}
	\label{fig:n1}
\end{figure*}

The chains of the hyperparameters are illustrated
in Fig.~\ref{fig:chain_n1}. Figs.~\ref{fig:gn_iter}
and~\ref{fig:gx_iter} represent the samples as function of the
iteration. We observe that, except for $N_D = 2$, all the chains
have the same behavior with the same convergence period. The $N_D = 2$
has slower convergence but reaches the same stationary distribution.

Figs.~\ref{fig:gn_times} and~\ref{fig:gx_times} represent the samples
as function of time (in seconds). The chains for RJ-PO and
GSGS with $N_D = 150$ have the same behavior. 
This result is obvious since both algorithms compute
almost the same number of gradients per iteration. That said, we see that for $N_D = 10$ and $N_D = 2$, the impact
on the convergence time is significant. The Tab.~\ref{tab:hyper} shows
some quantitative results. In particular the case $N_D = 10$ is ten
times faster than RJ-PO.

\begin{table}
	\centering
	\begin{tabular}{ccccc}
          & RJ-PO & 150 & 10 & 2 \\
          \hline
          $\widehat{\gn}$ & 0.9718 & 0.9694 & 0.9452 & 0.7078 \\
          $\widehat{\sigma_{\gn}}$ & 0.0063 & 0.0061 & 0.0066 & 0.3395 \\
          \hline
          $\widehat{\gx}$ & 1.07e-03 & 1.06e-03 & 1.95e-03 & 9.62e-03 \\
          $\widehat{\sigma_{\gx}}$ & 3.7e-05 & 1.7e-05 & 2.7e-05 & 6.2e-03 \\
          \hline
          loop [s.] & 4.4 & 2.4 & 0.2 & 0.1 \\
          total [s.] & 666 & 353 & 28 & 9
	\end{tabular}
	\caption{Hyper parameters values estimation with true $\gn = 1$.}
	\label{tab:hyper}
\end{table}

In addition, Tab.~\ref{tab:std10} shows a hyperparameter values
estimation with a higher noise level. Again the estimated values are
close and the $\gn$ parameter is correctly estimated.

\begin{table}
	\centering
	\begin{tabular}{cccc}
          &  RJ-PO & 10 & 2 \\
          \hline
          $\widehat{\gn}$  & 9.9e-03 & 9.9e-03 & 9.9e-03 \\
          $\widehat{\sigma_{\gn}}$ & 6.8e-05 & 4.8e-05 & 5.5e-05 \\
          \hline
          $\widehat{\gx}$  & 1.84e-03 & 3.28e-03 & 2.29e-03 \\
          $\widehat{\sigma_{\gx}}$ & 3.2e-04 & 7.0e-04 & 3.4e-05
	\end{tabular}

	\caption{Hyper parameters values with true $\gn = 0.01$.}
	\label{tab:std10}
\end{table}

\begin{figure}[htbp]
	\centering

	\subfloat[$\gn$ as iteration]{
		\includegraphics[width=0.2\textwidth]{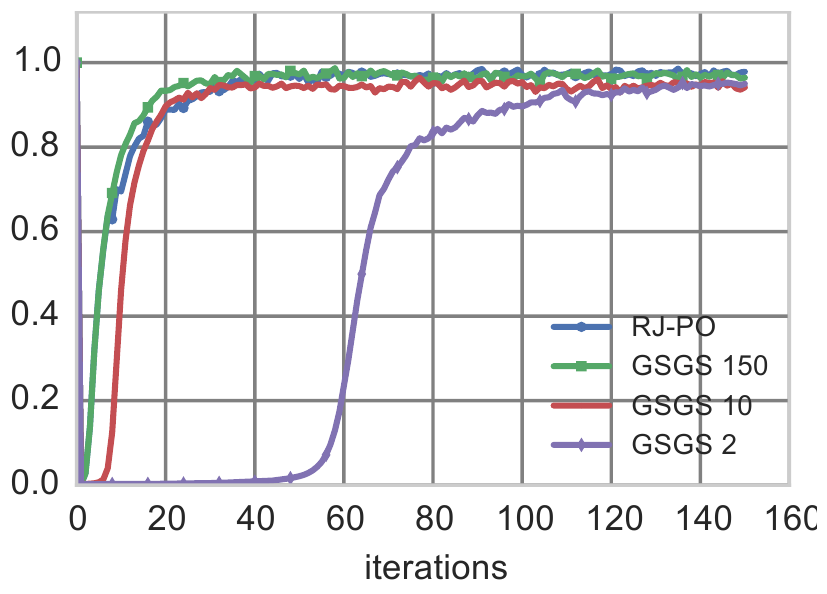}
		\label{fig:gn_iter}
	} %
	\subfloat[$\gn$ as time]{
		\includegraphics[width=0.2\textwidth]{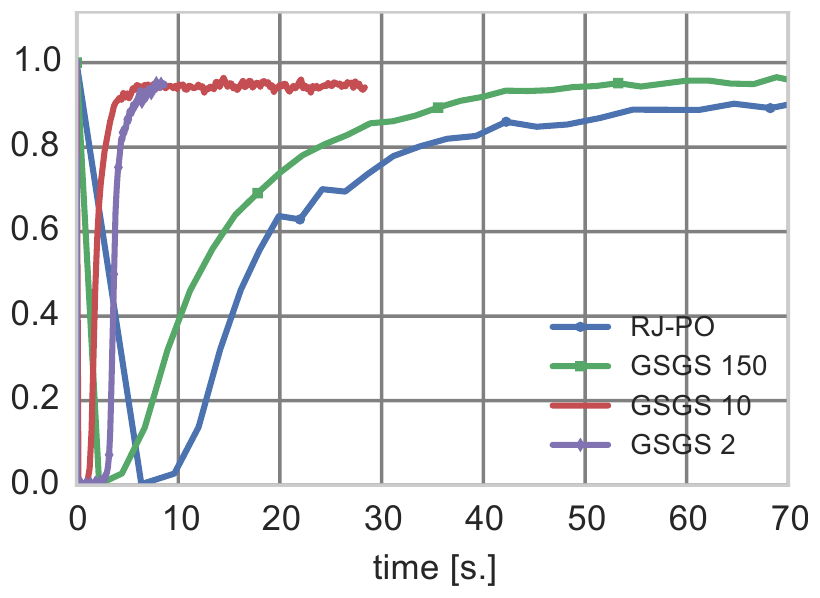}
		\label{fig:gn_times}
	}

	\subfloat[$\gx$ as iteration]{
		\includegraphics[width=0.2\textwidth]{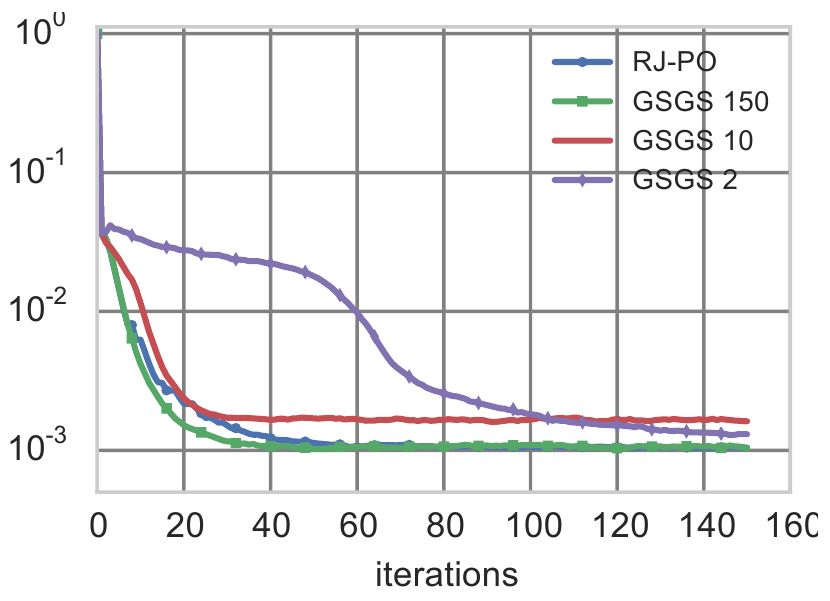}
		\label{fig:gx_iter}
	}
	\subfloat[$\gx$ as time]{
		\includegraphics[width=0.2\textwidth]{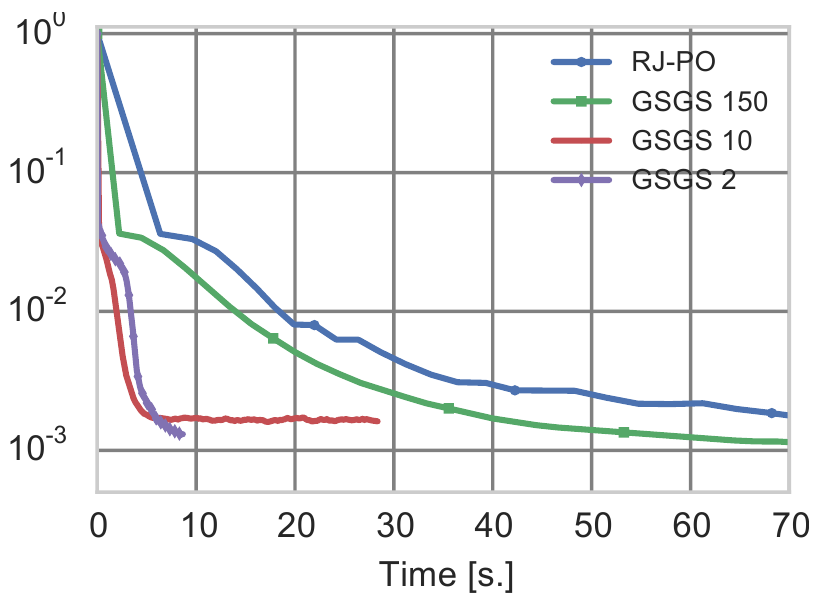}
		\label{fig:gx_times}
	}

	\caption{Chains of hyper parameters $\gx$ and $\gn$.}
	\label{fig:chain_n1}
\end{figure}

To illustrate the effect of the perturbation for good space
exploration, Fig.~\ref{fig:noper} shows the results when no
perturbations $\varepsilonb^{(t)}$ are done and with $N_D = 10$. In this case, the hypotheses of Proposition \ref{prop_proba_invariante} are no longer verified and those of Proposition~\ref{prop_irreductibilite} cannot be verified in practice. Moreover, the results show
that both the covariance and the hyperparameters are wrongly estimated. This effect leads to an over-regularized image. A possible explanation is that the conjugate directions of the GSGS explore in a privileged way the directions of small
variance (highest eigenvalues of $\Qb$).

\begin{figure}[htbp]
	\centering

	\subfloat[\textsc{pm}]{
		\includegraphics[height=0.19\textwidth]{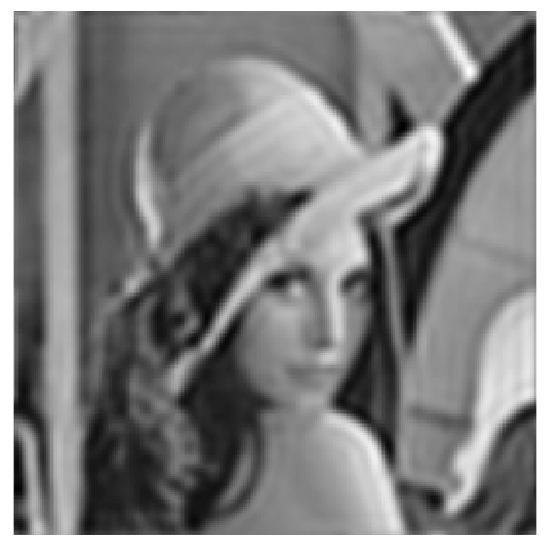}
	} %
	\subfloat[\textsc{psd}]{
		\includegraphics[height=0.19\textwidth]{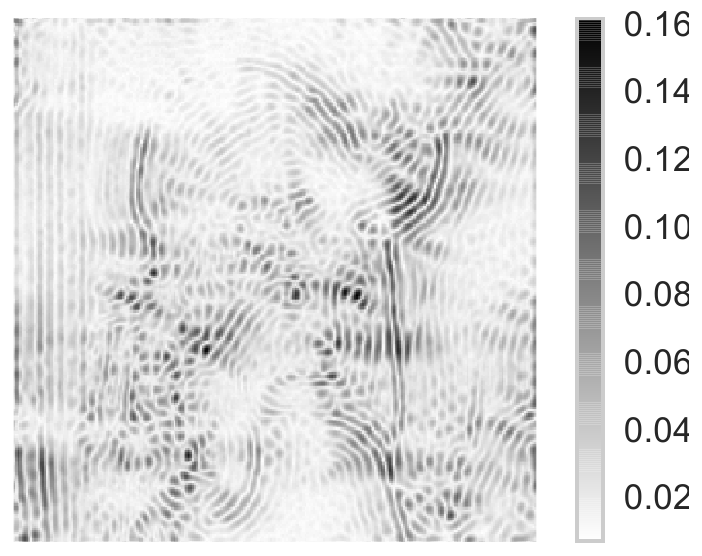}
	}

	\subfloat[$\gn$]{
		\includegraphics[width=0.2\textwidth]{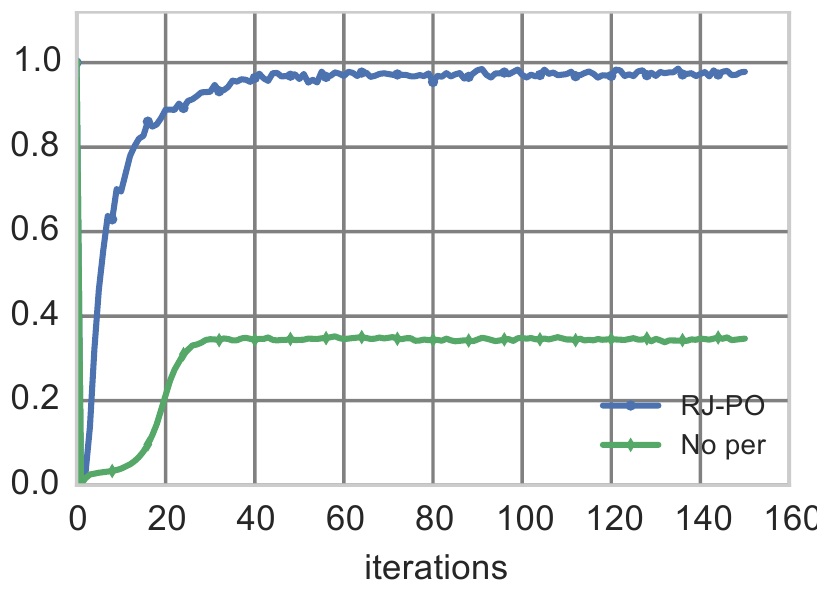}
		\label{fig:gn_iter_noper}
	} %
	\subfloat[$\gx$]{
		\includegraphics[width=0.2\textwidth]{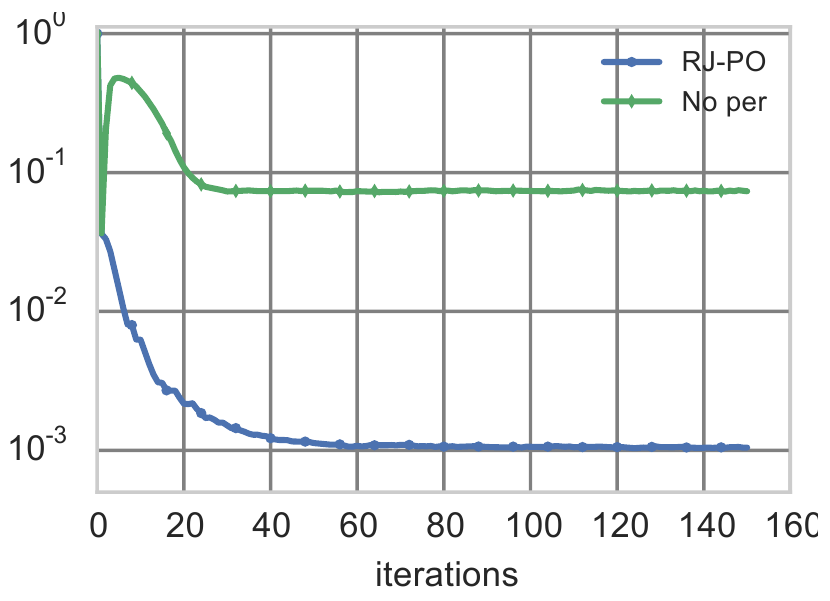}
		\label{fig:gx_iter_noper}
	}

	\caption{Results without perturbation and $N_D = 20$.}
	\label{fig:noper}
\end{figure}

Regarding the computational cost, all the presented algorithm are dominated by the cost of matrix-vector product $\Qb\xb$.  The cost thus depends on the specific problems and the structure of $\Qb$ in the same way than for conjugate gradient algorithm. For super-resolution problems, the cost of the matrix-vector product is almost equal to two discrete Fourier transforms of images. That said, the total number of matrix-vector is related to $N_D$ and the number of Gibbs iteration.

The main concluding comment is that the proposed algorithm allows a
great improvement in the convergence time of the Gibbs sampler while being
convergent to the true joint posterior law. However the speed
improvement can come with a bad covariance estimation if the
number $N_D$ of directions for the image $\xb$ is not sufficient.

\section{Conclusion}\label{sec:conclusion}

The handling of high-dimensional law, especially Gaussian, appears in
many linear inverse and estimation problems. With the growing interest
in ``Big Data'' and non stationary problems this task becomes
critical. Moreover, the uncertainty around the estimated
values, or the confidence interval, remains one of the difficult points combined with hyperparameter estimation for automatic method design.

The main contributions of this paper is (i) the proposition of a new algorithm in the class of the Gibbs samplers, able to address the case of high-dimensional Gaussian conditional laws, and (ii) the convergence proof of the algorithm. 
It relies on a random excursion along a small set of directions instead of handling with the high dimensional distribution. The directions are appropriately chosen according to the gradient of the potential.

This new algorithm is shown to be an efficient alternative to existing work like the PO-type algorithms: we ensure the theoretical convergence of the algorithm and, in some cases, we can show a drastic computing-time improvement.

The convergence of the algorithm is proved, provided that a random perturbation around the gradient direction is introduced. Even if in theory the only condition to ensure convergence is to choose a perturbation distribution supported in the whole space, it appears in practice that the results are very sensitive to the choice of the distribution. Moreover, the choice of the Gaussian distribution $\Nc(0,\Qb_\thetab)$ is the only case where the algorithm is more efficient than the PO and RJ-PO algorithm. The objective of our further work will be to better understand this high sensitivity to the choice of the perturbation distribution, that is, at this time, an open problem.

In further work the objective will be to study the convergence rate of the GSGS. In particular, the geometric ergodicity is an important property that ensures a fast convergence and allows us to give estimations of standard errors. The geometric ergodicity of Gibbs samplers has long been studied \cite{Roberts94} and a lot of results are shown in the Gaussian case \cite{Roberts97}, as well as for application in Bayesian hierarchical models \cite{Papaspiliopoulos08}, also in the case of joint Gaussian and Gamma distribution \cite{Hobert98, Johnson15}, the latter being close to our illustration example.

Also, one has to choose the number $N_D$ of mutually conjugate directions to sample at each iteration of the algorithm. In theory, this does not affect the convergence properties of the algorithm. As a perspective, one can propose an automatic choice of $N_D$, following the work in \cite{Gilavert15} for the RJ-PO.

The proposed algorithm is somewhat independent of the chosen
direction. The use of preconditioner to compute direction as in preconditioned conjugate gradient should improve the computational cost by an $N_D$ parameter smaller than at the present time. It depends, however, on each addressed problem.

This paper is
focused on linear conditionally Gaussian models. By use of hidden
variable, the algorithm should also be able to handle non Gaussian models that are still conditionally Gaussian.

\appendix

\subsection{Proof of Proposition \ref{prop_proba_invariante}}

\label{annexe_proba_invariante}
This appendix is devoted to prove Proposition \ref{prop_proba_invariante}. It is mainly inspired by the proofs presented in \cite{Levine06} (see also~\cite{Levine05,Latuszynski13}) for different random scan strategies in order to sample $p(\xb|\thetab)$. The only difference is that the random choice is not according to a set of coordinates of $\xb$ in the canonical basis, but according to a mutually conjugate set with respect to a current matrix $\Qb_\thetab$. Therefore the same arguments as detailed in \cite{Levine06} can be used to prove the irreducibility: if the support of the density $p(\varepsilonb)$ is $\Rbb^N$, all the directions can be explored in one step of the algorithm. Therefore any $\yb \in \Rbb^N$ can be reached in one step by taking, for example, $\db_1=\xtmun - \yb$, $\wt{\alpha}_1=1$, $\wt{\alpha}_n=0$, $n=2,\ldots,N_D$. Using classic continuity arguments, we can deduce that the probability of reaching any open ball $\Bc(\yb,r)$, centered in $\yb$ of radius $r$, conditional to any current point $\xt$, is strictly positive, which ensures the chain to be irreducible.

The rest of the proof focuses on the fact that $p(\xb,\thetab)$ is an invariant probability of the chain. We use the same arguments and notations of \cite{Levine06}.
Let $\xb \in \Rbb^N$ and a set $D$ of mutually conjugate directions with respect to a definite positive matrix $\Qb$. We decompose $\xb=(\xb_D,\xb_{\setminus D})$ which is always possible as explained in section \ref{sec_preliminary_results}.

Define $(\xb',\thetab') \in \Rbb^N \times \Theta$ a current point and $(\xb',\thetab') \in \Rbb^N \times \Theta$ the point obtained by Algorithm \ref{algo_pgsgs} with the transition Kernel:
\begin{equation*}
	P(\xb,\thetab|\xb',\thetab')= \pi(\thetab|\xb',\thetab') \pi(\xb_{D} | \xb_{\setminus D},\xb',\thetab) \delta (\xb_{\setminus D}- \xb'_{\setminus D})
\end{equation*}
with $\pi$ denoting any conditional probability and
$\delta$ is the Dirac function. The objective is to show that if $(\xb',\thetab')$ is distributed according to the joint distribution $p$, then $(\xb,\thetab)$ is also distributed according to $p$.

Let $A \subset \Rbb^N$ be a measurable set. The following lines are the result of the definition of the transition Kernel, the use of the general product rule, and of sequential integration with respect to $\thetab'$, $\xb'_D$ and $\xb'_{\setminus D}$:
\begin{eqnarray*}
&& \hspace*{-1cm}\Pbb((\xb,\thetab) \in A) \\
 & = & \int \unbb_{A}(\xb,\thetab) P(\xb,\thetab|\xb',\thetab') p(\xb',\thetab') \dD\xb \dD\thetab \dD\xb' \dD \thetab' \\
& = & \int \unbb_{A}(\xb,\thetab) \pi(\thetab|\xb',\thetab') \pi(\xb_{D} | \xb_{\setminus D},\xb',\thetab)\dots \\
& & \hspace*{1.1cm} \dots \delta
(\xb_{\setminus D}- \xb'_{\setminus D})p(\xb',\thetab') \dD\xb \dD\thetab \dD\xb' \dD \thetab'   \\
& = & \int \unbb_{A}(\xb,\thetab) p(\xb',\thetab) \pi(\xb_{D} | \xb_{\setminus D},\xb',\thetab)\dots \\
& & \hspace*{2.7cm} \dots \delta
(\xb_{\setminus D}- \xb'_{\setminus D}) \dD\xb \dD\thetab \dD\xb'  \\
& = & \int \unbb_{A}(\xb,\thetab) p(\xb'_{\setminus D},\thetab) \pi(\xb_{D} | \xb_{\setminus D},\xb'_{\setminus D},\thetab)\dots \\
& & \hspace*{2.7cm} \dots \delta
(\xb_{\setminus D}- \xb'_{\setminus D}) \dD\xb \dD\thetab \dD\xb'_{\setminus D}  \\
& = & \int \unbb_{A}(\xb,\thetab) p(\xb_{\setminus D},\thetab) \pi(\xb_{D} | \xb_{\setminus D},\thetab)\dD\xb \dD\thetab \\
& = & \int \unbb_{A}(\xb,\thetab) p(\xb,\thetab) \dD\xb \dD\thetab
\end{eqnarray*}
%

Hence the joint probability $p(\xb,\thetab)$ is an invariant probability of the Markov chain produced by Algorithm \ref{algo_pgsgs}.

\subsection{Proof of Proposition \ref{prop_irreductibilite}}
\label{annexe_irreductibilite}
This appendix is dedicated to prove Proposition \ref{prop_irreductibilite}. Let $(\xb^{(0)},\thetab^{(0)})\in \Rbb^N \times \Theta$ be a current point and $(\xt,\thetat)$ the point produced by the chain of Algorithm \ref{algo_pgsgs} at iteration $t$. The objective is to prove that for any non-negligible subset $A \subset \Rbb^N \times \Theta$, there is $T\ge0$ such as $\Pbb((\xb^{(T)},\thetab^{(T)})\in A|\xb^{(0)},\thetab^{(0)})>0$. Using the hypothesis H-2, it is sufficient to prove that for any non-negligible subset $A_x \in \Rbb^N$, there is $T\ge0$ such as:
\begin{equation}
	\Pbb(\xb^{(T)}\in A_x|\xb^{(0)},\thetab^{(0)})>0
\label{eq_proba_positif}
\end{equation}
Given $\xb^{(0)}$, we denote by $\thetab$ the corresponding element that respects conditions H-3. It is sufficient to prove the Proposition in the following framework:
\begin{enumerate}[{F}-1]
\item $\thetab^{(N+1)}=\thetab^{(N)} = \ldots = \thetab^{(0)}=\thetab$, \label{simple_theta}
\item $\mb_\thetab=0$, \label{simple_m}
\item $\Qb_\thetab=\text{diag}(q_1,\ldots,q_N)$ is diagonal. \label{simple_Q}
\end{enumerate}
Indeed, if we prove the inequality \eqref{eq_proba_positif} with fixed $\thetab$ for $N+1$ iterations, continuity arguments using conditions H-1 and H-2 will end the proof of the Proposition. The simplifications F-\ref{simple_m} and F-\ref{simple_Q} can be assumed by a change of variable $\yb^{(t)}=\xt-\mb_\thetab$ and by considering the basis of $\Rbb^N$ formed by the eigenvectors of $\Qb_\thetab$.

In this simplified framework, the chain of Algorithm \ref{algo_pgsgs} produces $\xt,~t=1,\ldots,N+1,$ such as:
\begin{equation*}
	\xt=(\Ib - \alpha^{(t)} \Qb_\thetab)(\Ib - \alpha^{(t-1)} \Qb_\thetab)\ldots (\Ib - \alpha^{(1)} \Qb_\thetab) \xb^{(0)},
\end{equation*}
with $\Ib$ the identity matrix in $\Rbb^N$ and, noting $\xb=(x_1,\ldots,x_N)\T$, we have, for $n=1,\ldots,N$:
\begin{equation}
\xt_n  = (1 - \alpha^{(t)} q_n)(1 - \alpha^{(t-1)} q_n)\ldots (1 - \alpha^{(1)} q_n) \xb^{(0)}_n.
\label{eq_polynome}
\end{equation}
The hypothesis H-3.2 ensures that $\xb^{(0)}_n\neq0$, $n=1,\ldots,N$, therefore we can assume without loss of generality that $\xb^{(0)}_n=1,~n=1,\ldots,N$, and equation  \eqref{eq_polynome} is, in this case:
\begin{equation}
	\xt_n  = (1 - \alpha^{(t)} q_n)(1 - \alpha^{(t-1)} q_n)\ldots (1 - \alpha^{(1)} q_n).
	\label{eq_polynome2}
\end{equation}

The following Lemma proves that any point in $\Rbb^N$ can be reached by the chain in $N+1$ iterations.
\begin{lemma}
For any $\yb \in \Rbb^N$, there is $\alpha=(\alpha^{(1)},\ldots,\alpha^{(N+1)})$ such as $\xb^{(N+1)}=\yb$, where $\xb^{(N+1)}$ is defined by \eqref{eq_polynome2} with $t=N+1$.
\end{lemma}
\begin{proof}
This can be done by interpreting it as an interpolation problem: given $\yb\in \Rbb^N$, the objective is to show that there is a polynomial $P_\alpha^{N+1}$ such as:
\begin{eqnarray}
	P_\alpha^{N+1}(q_n) & = & y_n, ~~ n=1,\ldots,N \label{eq_cons1}\\
	P_\alpha^{N+1}(0) & = & 1 \label{eq_cons2}
\end{eqnarray}
with $P_\alpha^{N+1}$ defined by the right hand side of \eqref{eq_polynome2} with $t=N+1$. The constraint \eqref{eq_cons2} is due to the specific form of $P_\alpha^{N+1}$. Also the fact that the parameters $\alpha^{(n)}$ must be real, implies that the polynomial $P_\alpha^{N+1}$ must have only real roots. It is well known that there is a polynomial of degree $N$ that respects \eqref{eq_cons1} and \eqref{eq_cons2}. Let  us denote by $Q$ such a polynomial. But the roots of $Q$ may be complex. However we can show that there is a polynomial of degree $N+1$ with real roots that respects the conditions \eqref{eq_cons1} and \eqref{eq_cons2}. Indeed, let us consider the polynomial $Q$ and a polynomial $R$ of degree $N+1$ such as $R(q_1)=R(q_2)=\ldots=R(q_N)=R(0)=0$. Therefore any polynomial $P_\tau=Q+\tau R$, $\tau \in \Rbb$, respects conditions \eqref{eq_cons1} and \eqref{eq_cons2}, and it is trivial to show that for $\tau^*$ sufficiently large, the polynomial $P_{\tau^*}$ has all its roots $r_n^*\in\Rbb,~n=1,\ldots,N$. Therefore, taking $P_\alpha^{N+1}=P_{\tau^*}$, \ie $\alpha^{(n)}=1/r_n^*$ ends the proof of the lemma.
\end{proof}
Using this lemma and the continuity of $P_\alpha^{N+1}$ with respect to $\alpha$, it is trivial to prove \eqref{eq_proba_positif} and then the Proposition.
\bibliographystyle{IEEEtran}

\bibliography{Biblio_GSGS}

\end{document}